\documentclass[11pt,letter,a4paper]{article}
\usepackage{amssymb}
\usepackage[normalem]{ulem}
\usepackage{natbib,comment}
\usepackage{graphicx}
\usepackage{rotating}
\usepackage[bookmarksopen,bookmarksnumbered,colorlinks,linkcolor={blue},citecolor={blue}]{hyperref}
\usepackage{theorem}
\usepackage{epsfig}
\usepackage{amsmath}
\usepackage{amsfonts}
\usepackage[singlespacing]{setspace}
\usepackage{times}

\newtheorem{theorem}{Theorem}
\newtheorem{assumption}{Assumption}

\newtheorem{corollary}[theorem]{Corollary}

\newtheorem{lemma}[theorem]{Lemma}

\newtheorem{proposition}[theorem]{Proposition}

\newenvironment{proof}[1][Proof]{\noindent\textbf{#1.} }{\ \rule{0.5em}{0.5em}}
\input{tcilatex}
\begin{document}

\author{Mogens Fosgerau\thanks{%
Technical University of Denmark; mfos@dtu.dk} \quad Emerson Melo\thanks{%
Indiana University; emelo@iu.edu} \quad Andr\'{e} de Palma\thanks{%
ENS Cachan, Universit\'{e} Paris-Saclay, CREST; andre.depalma@ens-cachan.fr.}
\quad Matthew Shum\thanks{%
California Institute of Technology; mshum@caltech.edu} }
\title{Discrete Choice and Rational Inattention: a General Equivalence Result%
\thanks{%
First draft: December 22, 2016. We thank Bob Becker, Marcus Berliant, Mark
Dean, Federico Echenique, Juan Carlos Escanciano, Filip Matejka, Paulo
Natenzon, Antonio Rangel, Ryan Webb, and Michael Woodford for useful
comments. Lucie Letrouit, Julien Monardo, and Alejandro Robinson Cortes
provided research assistance.}}
\date{\noindent {\today} \\
}
\maketitle

\begin{abstract}
This paper establishes a general equivalence between discrete choice and
rational inattention models. Matejka and McKay (2015, \emph{AER}) showed
that when information costs are modelled using the Shannon entropy function,
the resulting choice probabilities in the rational inattention model take
the multinomial logit form. By exploiting convex-analytic properties of the
discrete choice model, we show that when information costs are modelled
using a class of generalized entropy functions, the choice probabilities in
\emph{any} rational inattention model are observationally equivalent to some
additive random utility discrete choice model and vice versa.  Thus any additive random utility model can be given an
interpretation in terms of boundedly rational behavior.
This includes empirically relevant specifications  such as the probit and nested logit models.
%

\bigskip


\end{abstract}

\section{Motivation}

In many situations where agents must make decisions under uncertainty,
information acquisition is costly (involving pecuniary, time, or psychic
costs); therefore, agents may rationally choose to remain imperfectly
informed about the available options. This idea underlies the theory of
rational inattention, which has become an important paradigm for modeling
boundedly rational behavior in many areas of economics (Sims (2003, 2010)).
In this paper, our main contribution is to establish a general equivalence
between additive random utility discrete choice and rational inattention
models. Matejka and McKay (2015) showed that when information costs are
modelled using the Shannon entropy function, the resulting choice
probabilities in the rational inattention model take the multinomial logit
(MNL) form. In order for the rational inattention model to generate non-MNL
choice probabilities, we need to generalize the information cost function
beyond the Shannon entropy function assumed in much of the existing
literature. We do this by exploiting convex-analytic properties of the
additive random utility model to demonstrate a duality between discrete
choice and rational inattention models.\footnote{Throughout
this paper, we will use the terms ``additive random utility model'' and
``discrete choice model'' interchangeably.}

Specifically, we introduce a class of \emph{Generalized Entropy Rational
Inattention} (GERI) models.\footnote{%
This complements work by \citet{HebertWoodford16}, who also consider
generalizations of the information cost function.} In GERI models, the
information cost functions are constructed from a class of ``generalized
entropy'' functions; these generalized entropy functions are, essentially,
``dual'' to the class of random utility discrete choice models; precisely,
the generalized entropy functions are the convex conjugate functions to the
surplus functions in any arbitrary general additive random utility model.
Hence, GERI models naturally yield choice probabilities that can
equivalently be generated from general additive random utility models; the
resulting choice probabilities can take forms far beyond the multinomial
logit, including specifications such as nested logit, multinomial probit,
and so on, which are often employed in empirical work.

Importantly, these generalized entropy functions allow for random utility
models in which the random shocks are dependent across options; this
corresponds to information cost functions that exhibit information
spillovers across options with shared features, which may be reasonable in
many decision environments. In contrast, the multinomial logit model assumes
independent shocks; correspondingly, the Shannon entropy function precludes
information spillovers.

The paper is organized as follows. Section 2 presents insights into the
fundamental convex-analytic structure of the additive random utility
discrete choice model. Using this structure, we formulate a class of
generalized entropy functions and present key results about them. Section 3
introduces the rational inattention model. We show how the generalized
entropy functions can be used to define the information cost function in the
rational inattention model. Then we present the key result from this paper,
which establishes the equivalence between choice probabilities emerging from
the discrete choice model, and those emerging from the rational inattention
model based on the generalized entropy functions. Section 4 discusses an
example while Section 5 concludes.

\textbf{Notation:} Throughout this paper, for vectors $\mathbf{a}$ and $%
\mathbf{b}$, we use the notation $\mathbf{a}\cdot \mathbf{b}$ to denote the
vector scalar product $\sum_i a_ib_i$. $\Delta$ denotes the unit simplex in $%
\mathbb{R}^{N}$.

\section{Random utility models and generalized entropy functions\label%
{sec:dc}}

Consider a decision-maker (DM) making discrete choices among a set of $%
i=1,\ldots ,N$ options, where, for each option $i$, the utility is given by
\begin{equation}
u_{i}=\tilde{v}_{i}+\epsilon _{i},  \label{random utility}
\end{equation}%
where $\mathbf{\tilde{v}}=(\tilde{v}_{1},\ldots ,\tilde{v}_{N})$ is
deterministic and $\boldsymbol{\epsilon }=(\epsilon _{1},\ldots ,\epsilon
_{N})$ is a vector of random utility shocks. This is the classic additive
random utility framework pioneered by \citet{McFadden78}.

\begin{assumption}
\label{RUM_conditions} The random vector $\boldsymbol{\epsilon }=(\epsilon
_{1},\ldots ,\epsilon _{N})$ follows a joint distribution with finite means
that is absolutely continuous, independent of $\tilde{\mathbf{v}}$, and
fully supported on $\mathbb{R}^{N}$.
\end{assumption}

An important concept in this paper is the \emph{surplus function} of the
discrete choice model \citep[so named by ][]{McFadden81}, defined as
\begin{equation}  \label{SSF}
W(\mathbf{\tilde{v}}) = \mathbb{E}_{\boldsymbol{\epsilon}} (\max_{i} [\tilde{%
v}_i + \epsilon_i] ).
\end{equation}

Under Assumption \ref{RUM_conditions}, $W(\mathbf{\tilde{v}})$ is convex and
differentiable\footnote{%
The convexity of $W(\cdot )$ follows from the convexity of the max function.
Differentiability follows from the absolute continuity of $\boldsymbol{%
\epsilon }$.} and the choice probabilities coincide with the derivatives of $%
W(\mathbf{\tilde{v}})$:
\begin{equation*}
\frac{\partial W(\mathbf{\tilde{v}})}{\partial \tilde{v}_{i}}=q_{i}(\mathbf{%
\tilde{v}})\equiv \mathbb{P}\left( \tilde{v}_{i}+\epsilon _{i}\geq \tilde{v}%
_{j}+\epsilon _{j},\,\forall j\neq i\right) \,\,\,\mbox{for $i=1,\ldots,N$}
\end{equation*}%
or, using vector notation, $\mathbf{q}(\mathbf{\tilde{v}})=\nabla W(\mathbf{%
\tilde{v}})$. This is the Williams-Daly-Zachary theorem in the discrete
choice literature \citep{McFadden78, McFadden81}.

As a running example, we consider the familiar logit model. When the $%
\epsilon _{i}$'s are distributed i.i.d. across options $i$ according to the
type 1 extreme value distribution, then the resulting choice probabilities
take the familiar multinomial logit form: $q_{i}(\mathbf{\tilde{v}})=e^{%
\tilde{v}_{i}}/\sum_{j}e^{\tilde{v}_{j}}$. Assumption 1 above leaves the
distribution of the $\epsilon $'s unspecified, thus allowing for choice
probabilities beyond the multinomial logit case. Importantly, it
accommodates arbitrary correlation in the $\epsilon_{i}$'s across choices,
which is reasonable and realistic in applications.

We define a vector-valued function $\mathbf{H}(\cdot)=(H_1(\cdot),...,H_N(%
\cdot)):\mathbb{R}_+^N \mapsto \mathbb{R}_+^N$ as the gradient of the
exponentiated surplus, i.e.
\begin{equation}  \label{h function}
\mathbf{H}(e^{\tilde{\mathbf{v}}})= \nabla_{\mathbf{\tilde{v}}} \left( e^{W(%
\mathbf{\tilde{v}})}\right).
\end{equation}

From the differentiability of $W$ and the Williams-Daly-Zachary theorem it
follows that the choice probabilities emerging from any random utility
discrete choice model can be expressed in closed-form in terms of the $%
\mathbf{H}$ function as:\footnote{%
By direct differentiation of (\ref{h function}), and applying the
Williams-Daly-Zachary theorem, we have $q_i(\tilde{\mathbf{v}})=H_i(e^{W(%
\tilde{\mathbf{v}})})/e^{W(\tilde{\mathbf{v}})}$ for all $i$. Imposing the
summability restriction $\sum_i q_i(\tilde{\mathbf{v}})=1$ we have $\sum_i
H_i(e^{W(\tilde{\mathbf{v}})})=e^{W(\tilde{\mathbf{v}})}$ leading to Eq. (%
\ref{fdp1}).}
\begin{equation}
q_{i}(\tilde{\mathbf{v}})=\frac{H_{i}\left( e^{\tilde{\mathbf{v}}}\right) }{%
\sum_{j=1}^{N}H_{j}\left( e^{\tilde{\mathbf{v}}}\right) },\quad
\mbox{for
$i=1,\ldots,N$}.  \label{fdp1}
\end{equation}

For the multinomial logit case, the surplus function is $W(\mathbf{\tilde{v}}%
)=\log \left( \sum_{i=1}^N e^{\tilde{v}_{i}}\right) $, implying that $%
H_{i}(e^{\tilde{\mathbf{v}}})=e^{\tilde{v}_{i}}$. Thus, for this case Eq. (%
\ref{fdp1}) becomes the multinomial logit choice formula.

The function $\mathbf{H}$ is globally invertible (see Lemma \ref{h is
invertible} in the Appendix), and we introduce a function $\mathbf{S}$
defined as the inverse of $\mathbf{H}$,%
\begin{equation}
\mathbf{S}(\cdot)=\mathbf{H}^{-1}(\cdot).  \label{S_function}
\end{equation}%
For reasons that will be apparent below, we refer to $\mathbf{S}$ as a \emph{%
generator} function. There is a close relationship between the function $%
\mathbf{S}(\cdot )$ and the surplus function $W(\mathbf{\tilde{v}})$ of the
corresponding discrete choice model: as the next proposition establishes,
the surplus function $W(\cdot)$ and the generator function $\mathbf{S}(\cdot)
$ are related in terms of \emph{convex conjugate duality} \citep[ch.
12]{Rockafellar1970}.\footnote{%
For a convex function $g(\mathbf{x})$, its convex conjugate function is
defined as $g^*(\mathbf{y})=\max_{\mathbf{x}} \left\{\mathbf{x}\cdot \mathbf{%
y} - g(\mathbf{x})\right\}$, which is also convex. Roughly speaking, the
gradients (or sub-gradients, in case of non-differentiability) of $g(\mathbf{%
x})$ and $g^*(\mathbf{y})$ are inverse mappings to each other.}

\begin{proposition}[Convexity properties and generalized entropy functions]
\label{General_Entropy} Let assumption \ref{RUM_conditions} hold. Then:

\begin{description}
\item[(i)] The surplus function $W(\tilde{\mathbf{v}})$ is equal to
\begin{equation}
W(\tilde{\mathbf{v}})=\log \left( \sum_{i=1}^{N}H_{i}(e^{\tilde{\mathbf{v}}%
})\right) .  \label{surplusentropy}
\end{equation}

\item[(ii)] The convex conjugate function for the surplus function $W(%
\mathbf{\tilde{v}})$ is
\begin{equation*}
W^*(\mathbf{q})=\left\{
\begin{array}{ll}
\mathbf{q}\cdot \log \mathbf{S}(\mathbf{q}) & \text{$\mathbf{q}\in\Delta$}
\\
+\infty & \text{otherwise},%
\end{array}
\right.
\end{equation*}
where $\mathbf{S}(\cdot)$ is a generator function defined in (\ref%
{S_function}). We call the negative convex conjugate function $-W^*(\cdot)$
a {\bfseries generalized entropy function}.

\item[(iii)] The surplus function $W(\tilde{\mathbf{v}})$ is the convex
conjugate of $W^{\ast }(\mathbf{q})$, that is%
\begin{equation}
W(\mathbf{\tilde{v}})=\max_{\mathbf{q}\in \Delta }\left\{ \mathbf{q}\cdot
\mathbf{\tilde{v}}-W^{\ast }(\mathbf{q})\right\}  \label{surplusconjugate}
\end{equation}
and the RHS is optimized at the choice probabilities $\mathbf{q}(\mathbf{%
\tilde{v}})=\nabla W(\mathbf{\tilde{v}})$.
\end{description}
\end{proposition}

Parts (i) and (ii) establish a specific structure of the surplus function $W$
and its convex conjugate $W^*$; this is new in the literature on random
utility models, and may be of independent interest. We use this structure to
define the class of \emph{generalized entropy} functions. To see how this
works, consider again the multinomial logit model, for which $\mathbf{H}$ is
the identity, implying that the corresponding generator function $\mathbf{S}(%
\mathbf{q})=\mathbf{q}$ is also just the identity. Then by Proposition~\ref%
{General_Entropy}(ii), the negative convex conjugate function is $-W^{\ast }(%
\mathbf{q})=-\mathbf{q}\cdot \log \mathbf{q}=-\sum_i q_i \log q_i$, which is
just the \citet{Shannon1948} entropy function.

Generalizing from this, Proposition~\ref{General_Entropy}(ii) shows how the
conjugate function for any discrete choice model can be generated by the
function $\mathbf{S} $. Therefore we refer to the negative conjugate
function $-W^{\ast }(\mathbf{q})=-\mathbf{q}\cdot \log \mathbf{S}(\mathbf{q}%
)=-\sum_i q_i \log S_i(\mathbf{q})$ for any general discrete choice model as
a \emph{generalized entropy} function. Comparing the generalized and Shannon
entropies, the former allows for cross-effects, in the sense that the choice
probability for option $j$, $q_j$, enters the entropic term for option $i$, $%
S_i(\mathbf{q})$. As we will see below, these cross-effects allow for
information spillovers when we use these generalized entropy functions to
construct rational inattention models.

Proposition~\ref{General_Entropy}(iii) provides an alternative
representation of the surplus function from a random utility model, in
addition to Eq. (\ref{SSF}). It illustrates a close connection between $-W^*(%
\mathbf{q})$ and the joint distribution of $\boldsymbol{\mathbf{\epsilon }}$%
, the random utility shocks, which aids interpretation of the generalized
entropy function. Specifically, 
Eq. (\ref{SSF}) implies that the surplus function can be written as
\begin{equation*}
W(\tilde{\mathbf{v}})=\sum_{i=1}^{N}q_{i}(\tilde{\mathbf{v}})(\tilde{v}_{i}+%
\mathbb{E}(\epsilon _{i}|u_{i}\geq u_{j},j\neq i)).
\end{equation*}
Combining this with (\ref{surplusconjugate}), we obtain an alternative
expression for the generalized entropy function, as a choice
probability-weighted sum of expectations of the utility shocks $\boldsymbol{%
\mathbf{\epsilon}}$:\footnote{%
See \citet{ChiongEtAl2016}. Additionally, we conjecture that $\log S_{i}(%
\mathbf{q})=-\mathbb{E}[\epsilon _{i}|u_{i}\geq u_{j},j\neq i]\quad %
\mbox{for $i=1,\ldots,N$}$. For the multinomial logit case, corresponding to
$\mathbf{S}(\mathbf{q})=\mathbf{q}$, \citet{McFadden78} showed that $\gamma
-\log q_{i}=\mathbb{E}[\epsilon _{i}|u_{i}\geq u_{j},j\neq i]$, for $\gamma $
being Euler's constant.}
\begin{equation*}
-W^{\ast }(\mathbf{q})=\sum_{i}q_{i}\mathbb{E}[\epsilon _{i}|u_{i}\geq
u_{j},j\neq i].
\end{equation*}%
In this way, different distributions for the utility shocks $\boldsymbol{%
\mathbf{\epsilon}}$ in the random utility model will imply different
generalized entropy functions, and vice versa.

We conclude this section enumerating some properties of the generator
functions $\mathbf{S}(\cdot)$, which will be important in what follows.

\begin{proposition}[Properties of the generator functions]
\label{Generator} Let assumption \ref{RUM_conditions} hold. Then the vector
valued-function $\mathbf{S}(\cdot)$ defined by (\ref{S_function}) satisfies
the following conditions:

\begin{description}
\item[(i)] $\mathbf{S}$ is continuous and homogenous of degree 1.

\item[(ii)] $\mathbf{q}\cdot\log\mathbf{S}(\mathbf{q})$ is convex.

\item[(iii)] $\mathbf{S}$ is differentiable 
with :
\begin{equation*}
\sum_{i=1}^{N}q_{i}\frac{\partial \log S_{i}(\mathbf{q})}{\partial q_{k}}=1
,k\in \{1,\ldots ,N\},
\end{equation*}%
where
$\mathbf{q}$ is a probability vector with $0<q_i<1$ for all $i$.
\end{description}
\end{proposition}


The possibility of zero choice probabilities will play a role in what
follows. We impose an additional regularity assumption on the generator
functions $\mathbf{S}$.

\begin{assumption}
\label{RUM_regularity} Let $\mathbf{q}$ be a probability vector. Then $%
q_{i}=0$ if and only if $S_{i}(\mathbf{q})=0$.
\end{assumption}

This assumption is satisfied for the generator functions $\mathbf{S}$
corresponding to many familiar additive random utility models, including the
multinomial logit and the nested logit models.\footnote{%
In fact, the necessity part of Assumption \ref{RUM_regularity} arises
immediately from the results in this section. As $\tilde{v}_{i}\rightarrow
-\infty $, $q_{i}(\tilde{\mathbf{v}})\rightarrow 0$, which by (\ref{fdp1})
implies that $H_{i}\left( e^{\tilde{\mathbf{v}}}\right) \rightarrow 0$.
Then, since $\log \mathbf{S}(\mathbf{q}(\tilde{\mathbf{v}}))=\tilde{\mathbf{v%
}}-\log \sum_{j}H_{j}(e^{\tilde{\mathbf{v}}})$, we have $\log
S_{1}(q)\rightarrow -\infty $ (by homogeneity of $\mathbf{H}$, we may
suppose that $\log \sum_{j}H_{j}(e^{\tilde{\mathbf{v}}})$ is a constant).}

\section{Rational inattention\label{s3}}

We now introduce the rational inattention model. The decision maker is again
presented with a group of $N$ options, from which he must choose one. Each
option has an associated payoff $\mathbf{v}=(v_{1},...,v_{N})$, but in
contrast to the additive random utility model, the vector of payoffs is
unobserved by the DM. Instead, the DM considers the payoff vector $\mathbf{V}
$ to be random, taking values in a set $\mathcal{V}\subset \mathbb{R}^{N}$;
for simplicity, we take $\mathcal{V}$ to be finite. The DM possesses some
prior knowledge about the available options, given by a probability measure $%
\mu (\mathbf{v})=\mathbb{P}(\mathbf{V}=\mathbf{v})$.

The DM's choice is represented as a random action $\mathbf{A}$ that is a
canonical unit vector in $\mathbb{R}^{N}$. The payoff resulting from the
action is $\mathbf{V}\cdot \mathbf{A}$, namely that value of the entry in $%
\mathbf{V}$ indicated by the action $\mathbf{A}$. The problem of the
rationally inattentive DM is to choose the conditional distribution $\mathbb{%
P} (\mathbf{A}|\mathbf{V})$, balancing the expected payoff against the cost
of information.

Denote an action by $i$ and write $p_{i}(\mathbf{v})$ as shorthand for $%
\mathbb{P} \left( \mathbf{A}=i|\mathbf{V}=\mathbf{v}\right) $. Denote also
the vector of choice probabilities conditional on $\mathbf{V}=\mathbf{v}$ by
$\mathbf{p}(\mathbf{v})=(p_{1}(\mathbf{v}),\dots ,p_{N}(\mathbf{v}))$, and
let $\mathbf{p}(\cdot )=\{\mathbf{p}(\mathbf{v})\}_{\mathbf{v}\in\mathcal{V}%
} $ denote the collection of conditional probabilities. The DM's strategy is
a solution to the following variational problem:
\begin{equation}
\max_{\mathbf{p}(\cdot )}\left( \mathbb{E}\left( \mathbf{V}\cdot \mathbf{A}%
\right) -\mbox{information cost}\right) .  \label{Program1}
\end{equation}

The previous literature has used the Shannon entropy to specify the
information cost, which connects the rational inattention model to the
multinomial logit model. We review these results in the next Section \ref%
{sec:shannon_and_MNL}. Then in Section \ref{sec:genent_RI} we introduce
generalized entropy to the problem. This connects the rational inattention
model to general additive random utility models.

\subsection{Shannon entropy and multinomial logit\label{sec:shannon_and_MNL}}

The key element in the program above is the cost of information. Much of the
previous literature has utilized the mutual (Shannon) information between
payoffs $\mathbf{V}$ and the actions $\mathbf{A}$ to measure the information
costs. Denote the Shannon entropy by $\Omega (\mathbf{q})=-\mathbf{q}\cdot
\log \mathbf{q}$. Denote also the unconditional choice probabilities by $%
p_{i}^{0}=\mathbb{E}p_{i}(\mathbf{V})$ and $\mathbf{p}^{0}=(p_{1}^{0},\dots
,p_{N}^{0})$. Then the mutual (Shannon) information between $\mathbf{V}$ and
$\mathbf{A}$ may be written as%
\begin{eqnarray}
\kappa (\mathbf{p}\left( \cdot \right) ,\mu ) &=&\Omega (\mathbb{E}(\mathbf{p%
}(\mathbf{V})))-\mathbb{E}(\Omega (\mathbf{p}(\mathbf{V})))
\label{MutualEntropy} \\
&=&-\sum_{i=1}^{N}p_{i}^{0}\log p_{i}^{0}+\sum_{\mathbf{v}\in \mathcal{V}%
}\left( \sum_{i=1}^{N}p_{i}(\mathbf{v})\log p_{i}(\mathbf{v})\right) \mu (%
\mathbf{v}).
\end{eqnarray}

Accordingly, we can specify the information cost as $\lambda \kappa (\mathbf{%
p},\mu )$ where $\lambda >0$ is the unit cost of information. As the
distribution of payoffs is unspecified, we may take $\lambda =1$ at no loss
of generality. The choice strategy of the rationally inattentive DM is the
distribution of the action $\mathbf{A}$ conditional on the payoff $\mathbf{V}
$\textbf{\ }that maximizes the expected payoff less the cost of information,
which is the solution to the optimization problem
\begin{equation}
\max_{\mathbf{p}\left( \cdot \right) }\left\{ \sum_{\mathbf{v}\in \mathcal{V}%
}\left( \sum_{i=1}^{N}v_{i}p_{i}(\mathbf{v})\right) \mu (\mathbf{v})-\kappa (%
\mathbf{p}\left( \cdot \right) ,\mu )\right\}  \label{Program1a}
\end{equation}%
subject to
\begin{equation}
p_{i}(\mathbf{v})\geq 0\ \text{for all $i$},\quad\sum_{i=1}^{N}p_{i}(\mathbf{%
v})=1.  \label{constrain2}
\end{equation}%
Solving this, the DM finds conditional choice probabilities
\begin{equation}
p_{i}(\mathbf{v})=\frac{p_{i}^{0}e^{v_{i}}}{\sum_{j=1}^{N}p_{j}^{0}e^{v_{j}}}%
\quad \mbox{for $i=1,\ldots,N$},  \label{Generalized_Logit}
\end{equation}%
that satisfy $p_{i}^{0}=\mathbb{E}p_{i}(\mathbf{V})$. It is an important
feature of the rational inattention model that some $p_{i}^{0}$ may be zero,
in which case the corresponding $p_{i}\left( \mathbf{v}\right) $ are also
zero. Then the rational inattention model implies the formation of a \emph{%
consideration set}, comprising those options that have strictly positive
probability of being chosen (cf. \citet{CaplinDeanLeahy16}).

Under the convention that $\log 0=-\infty $ and $\exp \left( -\infty \right)
=0$, we may rewrite (\ref{Generalized_Logit}) as
\begin{equation*}
p_{i}(\mathbf{v})=\frac{e^{v_{i}+\log p_{i}^{0}}}{\sum_{j=1}^{N}e^{v_{j}+%
\log p_{j}^{0}}}=\frac{e^{\tilde{v}_{i}}}{\sum_{j=1}^{N}e^{\tilde{v}_{j}}},
\end{equation*}%
where $\tilde{v}_{i}=v_{i}+\log p_{i}^{0}$. This may be recognized as a
multinomial logit model in which the payoff vector $\mathbf{\tilde{v}}$ is $%
\mathbf{v}$ shifted by $\log \mathbf{p}^{0}$. For options that are not in
the consideration set, the shifted payoff is $\tilde{v}_{i}=-\infty $. From
the perspective of the multinomial logit model these options have zero
probability of maximizing the random utility (\ref{random utility}) and they
have effectively been eliminated from the model.%
%

\subsection{The Generalized Entropy Rational Inattention (GERI) model\label%
{sec:genent_RI}}

In this paper we generalize the preceding equivalence result between
rational inattention and multinomial logit. To achieve that, we replace the
Shannon entropy by the generalized entropy introduced in Section \ref{sec:dc}
above. Since each generalized entropy implies a corresponding discrete
choice model (Proposition 2), it turns out that each RI model with an
information cost derived from a generalized entropy will generate choice
probabilities consistent with a corresponding discrete choice model
(Proposition 4 below); this implies that \emph{any} additive random utility
discrete choice model can be microfounded by a corresponding rational
inattention model, thus generalizing substantially the results in the
previous section.

We begin by generalizing the rational inattention framework described above,
using generalized entropy in place of the Shannon entropy. Specifically, we
let $\mathbf{S}$ be the entropy generator corresponding to some additive
random utility model satisfying Assumptions \ref{RUM_conditions} and \ref%
{RUM_regularity}, and define $\Omega _{\mathbf{S}}\left( \mathbf{p}\right) =-%
\mathbf{p}\cdot \log \mathbf{S}\left( \mathbf{p}\right) $ as the
corresponding generalized entropy. We define accordingly a general
information cost by
\begin{eqnarray}
\kappa _{\mathbf{S}}\left( \mathbf{p}\left( \cdot \right) ,\mathbf{\mu }%
\right) &=&\Omega _{\mathbf{S}}\left( \mathbb{E}\mathbf{p}(\mathbf{V}%
)\right) -\mathbb{E}\Omega _{\mathbf{S}}\left( \mathbf{p}(\mathbf{V})\right)
\label{GeneralizedEntropy} \\
&=&-\mathbf{p}^{0}\cdot \log \mathbf{S}\left( \mathbf{p}^{0}\right)
+\sum\limits_{\mathbf{v}\in \mathcal{V}}\left[ \mathbf{p}\left( \mathbf{v}%
\right) \cdot \log \mathbf{S}\left( \mathbf{p}\left( \mathbf{v}\right)
\right) \right] \mu \left( \mathbf{v}\right) .  \notag
\end{eqnarray}

A \emph{Generalized Entropy Rational Inattention}{\ (GERI)} model describes
a DM who chooses the collection of conditional probabilities $\mathbf{p}%
\left( \cdot \right) =\{\mathbf{p}(\mathbf{v})\}_{\mathbf{v}\in \mathcal{V}}$
to maximize his expected payoff less the general information cost%
\begin{equation}
\max_{\mathbf{p}\left( \cdot \right) }\sum_{\mathbf{v}\in \mathcal{V}}\left(
\sum_{i=1}^{N}v_{i}p_{i}(\mathbf{v})\right) \mu (\mathbf{v})-\kappa _{%
\mathbf{S}}(\mathbf{p}\left( \cdot \right) ,\mu ).  \label{GERI_program}
\end{equation}

The following proposition characterizes the optimal solution to the GERI
model.

\begin{proposition}
\label{GERI_Solution} The solution to the GERI\ model:

\begin{description}
\item[(i)] The unconditional probabilities satisfy the fixed point equation%
\begin{equation}
\mathbf{p}^{0}=\mathbb{E}\left( \frac{\mathbf{H}\left( e^{\mathbf{V}+\log
\mathbf{S}\left( \mathbf{p}^{0}\right) }\right) }{\sum_{j=1}^{N}H_{j}\left(
e^{\mathbf{V}+\log \mathbf{S}\left( \mathbf{p}^{0}\right) }\right) }\right) .
\label{GERI_fixedpoint}
\end{equation}

\item[(ii)] The conditional probabilities are given in terms of the
unconditional probabilities by%
\begin{equation}
p_{i}\left( \mathbf{v}\right) =\frac{H_{i}\left( e^{\mathbf{v}+\log \mathbf{S%
}\left( \mathbf{p}^{0}\right) }\right) }{\sum_{j=1}^{N}H_{j}\left( e^{%
\mathbf{v}+\log \mathbf{S}\left( \mathbf{p}^{0}\right) }\right) }.
\label{GERI_choices}
\end{equation}

\item[(iii)] The optimized value of (\ref{GERI_program}) is
\begin{equation*}
\mathbb{E}\log \sum_{j=1}^{N}H_{j}\left( e^{\mathbf{V}+\log \mathbf{S}\left(
\mathbf{p}^{0}\right) }\right) =\mathbb{E}W\left( \mathbf{V}+\log \mathbf{S}%
\left( \mathbf{p}^{0}\right) \right) .
\end{equation*}
\end{description}
\end{proposition}

Part (i) of the proposition shows that the solution of the GERI model
involves a fixed point problem; in what follows, we assume that a solution
exists. Part (iii) illustrates the close connection between convex analysis
and the GERI problem. To see this, note that the GERI information cost
function may be written as
\begin{equation}
\kappa _{\mathbf{S}}(\mathbf{p}(\cdot),\mu )=-W^{\ast }(\mathbf{p}^{0})+%
\mathbb{E}W^{\ast }(\mathbf{p}(\mathbf{V})).  \label{Kconjugate}
\end{equation}%
Hence, given $\mathbf{p}^{0}$, the conditional choice probabilities $\mathbf{%
p}(\mathbf{v})$
can be generated, for each $\mathbf{v}\in \mathcal{V}$, by the problem
\begin{equation}
\max_{\mathbf{p}\left( \mathbf{v}\right) \in \Delta }\left\{ \mathbf{p}(%
\mathbf{v})\cdot (\mathbf{v}+\log \mathbf{S}(\mathbf{p}^{0}))-W^{\ast }(%
\mathbf{p}(\mathbf{v}))\right\} ,  \label{ratconjugate}
\end{equation}%
the optimized value of which, by Proposition~\ref{General_Entropy}(iii), is
\begin{equation}
W(\mathbf{v}+\log \mathbf{S}(\mathbf{p}^{0})),\quad \text{for each $\mathbf{v%
}\in \mathcal{V}$}  \label{program3}
\end{equation}%
corresponding to Proposition~\ref{GERI_Solution}(iii).

It is worth remarking that some of the optimal unconditional choice
probabilities may be zero. For these options, the corresponding conditional
choice probabilities will also be zero for all $\mathbf{v}$.\footnote{%
To see this, consider the solution to the GERI problem given in Eq. (\ref%
{GERI_choices}) and define $\mathbf{\tilde{v}}=\mathbf{v}+\log\mathbf{S}(%
\mathbf{p}^0)$. Let $p_i^0=0$. Then by assumption \ref{RUM_regularity} it
follows that $\log S_i(\mathbf{p}^0)=-\infty$, or equivalently, $\tilde{v}%
_i\longrightarrow -\infty$ and hence $p_i(\mathbf{v})=0$ for all $\mathbf{v}%
\in \mathcal{V}$.} The rational inattention model then also describes the
formation of consideration sets, i.e. the set of options that are chosen
with positive probability.\footnote{%
Because of the possibility of zero choice probabilities for some options,
GERI models can also generate failures of the ``regularity'' property
(adding an option to a choice set cannot increase the choice probability for
any of the original choices). See section~\ref{nonregularity} in the
Appendix for an example.}

While Proposition \ref{GERI_Solution} does not characterize explicitly the
optimal consideration set emerging from a GERI\ model, the following
corollary describes one important feature that it has, namely that it
excludes options that offer the lowest utility in all states of the world.

\begin{corollary}
\label{thm:consideration} For some option $a$, and for all $\mathbf{v}\in
\mathcal{V}$, let $v_{a}\leq v_{i}$ for all $i\neq a$, and assume that the
inequality is strict with positive probability. Then $p_{a}^{0}=0$ (that is,
option $a$ is not in the optimal consideration set).
\end{corollary}

For the special case of Shannon entropy (when $\mathbf{S}$ is the identity
function), the result can be strengthened even further. Corollary \ref%
{thm:considerationShannon} in the Appendix shows that in that case, an
option that is dominated by another option in all states of the world will
not be in the optimal consideration set.

\subsection{Equivalence between discrete choice and rational inattention
\label{s4}}

We now establish the central result of this paper, namely the equivalence
between additive random utility discrete choice models and rational
inattention models. In particular, we show that the choice probabilities
generated by a GERI model lead to the same choice probabilities as a
corresponding additive random utility model and vice versa.

Combining the choice probabilities $p_{i}(\mathbf{v})$ in (\ref{GERI_choices}%
) from the GERI\ model and the choice probabilities $q_{i}(\tilde{\mathbf{v}}%
)$ in (\ref{fdp1}) from the additive random utility model, we find that if
payoffs are related by
\begin{equation}
\tilde{v}_{i}=v_{i}+\log S_{i}(\mathbf{p}^{0})\quad
\mbox{for
$i=1,\ldots,N$},  \label{perturbed_valuation}
\end{equation}%
then the two models yield the same choice probabilities
\begin{equation*}
p_{i}(\mathbf{v})=\frac{H_{i}(e^{\mathbf{v}+\log \mathbf{S(}\mathbf{p}^{%
\mathbf{0}})})}{\sum_{j=1}^{N}H_{j}(e^{\mathbf{v}+\log \mathbf{S(}\mathbf{p}%
^{\mathbf{0}})})}=\frac{H_{i}(e^{\tilde{\mathbf{v}}})}{%
\sum_{j=1}^{N}H_{j}(e^{\tilde{\mathbf{v}}})}=q_{i}(\tilde{\mathbf{v}}).
\end{equation*}

Given a GERI\ model with payoffs $\mathbf{v}\in \mathcal{V}$ and
unconditional choice probabilities $\mathbf{p}^{0}$, we may then use (\ref%
{perturbed_valuation}) to construct deterministic utility components $\tilde{%
\mathbf{v}}$ for the additive random utility model. If the GERI model has
some zero unconditional choice probabilities $p_{i}^{0}$, then Assumption %
\ref{RUM_regularity} ensures that $p_{i}(\mathbf{v})=0$ if and only if $%
q_{i}(\tilde{\mathbf{v}})=0$. The additive random utility model that
corresponds to the GERI\ model is then an extended additive random utility
model in which some deterministic utility components are minus infinity.

Conversely, given an additive random utility model with flexible generator $%
\mathbf{S}$ and a prior distribution $\tilde{\mu}$ of the deterministic
utility components $\tilde{\mathbf{v}}\in \mathcal{\tilde{V}}$, define $%
\mathbf{p}^{0}=\mathbb{E}\mathbf{q}(\tilde{\mathbf{v}})$ and note that all $%
p_{i}^{0}>0$. Then define $\mathbf{v}$ using (\ref{perturbed_valuation}) and
define similarly $\mu $ and $\mathcal{V\ }$using the same location shift $%
\log \mathbf{S}(\mathbf{p}^{0})$. By the same argument as before, the GERI\
model with payoffs $\mathbf{v}\in \mathcal{V}$, prior $\mu $ and flexible
generator $\mathbf{S}$ for the generalized entropy has the same choice
probabilities as the additive random utility model.

Hence, we have shown the following proposition.

\begin{proposition}
\label{equivalence} For every additive random utility discrete choice model
and every prior distribution on $\mathcal{\tilde{V}}$, there is an
equivalent GERI\ model with a prior distribution on $\mathcal{V}$, where $%
\mathcal{V}$ is equal to $\mathcal{\tilde{V}}$ up to a location shift.

Conversely, every GERI model is equivalent to an additive random utility
discrete choice model in which the utility components for options chosen
with zero probability are set to minus infinity.
\end{proposition}

In Section~\ref{nestedlogitexample}, we will apply this proposition to study
a GERI model in which the choice probabilities are equivalent to those from
a nested logit discrete choice model.

\subsection{Additional properties of generalized entropy cost functions}

We have shown that the generalized rational inattention model is always
equivalent to an additive random utility model and conversely that the
generalized rational inattention model may provide a boundedly rational
foundation for any additive random utility model. The key to this result is
the generalization of the information cost function $\kappa _{\mathbf{S}}(%
\mathbf{p}\left( \cdot \right) ,\mu )$ using generalized entropy as defined
in Eq. (\ref{GeneralizedEntropy}). It is then natural to ask whether $\kappa
_{\mathbf{S}}(\mathbf{p}\left( \cdot \right) ,\mu )$ has the properties that
one would desire for an information cost. In this section we show that $%
\kappa _{\mathbf{S}}(\mathbf{p}\left( \cdot \right) ,\mu )$ does indeed
possess two reasonable and desirable properties of cost functions that have
been discussed in the existing literature (cf. \citet{deOliveira2015}, %
\citet{HebertWoodford16}), thus providing normative support for the GERI
framework.

First, when $\mathbf{A}$ and $\mathbf{V}$ are independent, then the action $%
\mathbf{A}$ carries no information about the payoff $\mathbf{V}$. In that
case the information cost should be zero, i.e.

\begin{description}
\item {\bfseries\emph{Independence}}. \emph{If $\mathbf{A}$ and $\mathbf{V}$
are independent, then $\kappa_\mathbf{S}(\mathbf{p}(\cdot),\mu)=0$.}
\end{description}


Second, the mutual Shannon information $\kappa (\mathbf{p}\left( \cdot
\right) ,\mu )$ is a convex function of $\mathbf{p}$. This is useful as it
ensures a unique solution to the problem of the rationally inattentive DM.
We show that the information cost $\kappa _{\mathbf{S}}(\mathbf{p}\left(
\cdot \right) ,\mu )$ has a slightly weaker property, namely that it is
convex on sets where $\mathbb{E}\mathbf{p}(\mathbf{V})$ is constant.

\begin{description}
\item {\bfseries\emph{Convexity}}. \emph{For a given $\mu $, the information
cost function }$\kappa _{\mathbf{S}}(\mathbf{p}\left( \cdot \right) ,\mu )$%
\emph{\ is convex on any set of choice probabilities vectors satisfying ${%
\left\{ \mathbf{p}:\mathcal{V}\mapsto \Delta|\ \mathbb{E}\mathbf{p}(%
\mathbf{V})=\mathbf{\hat{p}}\right\} }$.}
\end{description}

The mutual Shannon information $\kappa (\mathbf{p}\left( \cdot \right) ,\mu
) $ satisfies these two properties. The next proposition establishes that
the information cost defined in (\ref{GeneralizedEntropy}) using the
generalized entropy functions also satisfies these properties.

\begin{proposition}
\label{niceproperties} The information cost defined in Eq. (\ref%
{GeneralizedEntropy}) satisfies the independence and convexity conditions.
\end{proposition}


\section{Example: The nested logit GERI model}

\label{nestedlogitexample}

From an applied point of view, an important implication of Proposition \ref%
{equivalence} is that it allows us to formulate rational inattention models
that have complex substitution patterns, beyond the multinomial logit case.
In this example, we consider a GERI model with an information cost function
derived from a nested logit discrete choice model. The nested logit choice
probabilities are consistent with a discrete choice model in which the
utility shocks $\mathbf{\epsilon }$ are jointly distributed in the class of
generalized extreme value distributions. Among applied researchers, the
nested logit model is often preferred over the multinomial logit model
because it allows some products to be closer substitutes than others, thus
avoiding the \textquotedblleft red bus/blue bus\textquotedblright\ criticism.%
\footnote{%
See, for instance, \citet[Chap. 2]{Maddalabook}, and %
\citet{AndersondePalmaThisse1992}.}

We partition the set of options $i\in \left\{ 1,\ldots ,N\right\} $ into
mutually exclusive nests, and let $g_{i}$ denote the nest containing option $%
i$. Let $\zeta _{g_{i}}\in (0,1]$ be nest-specific parameters. For a
valuation vector $\mathbf{\tilde{v}}$, the nested logit choice probabilities
are given by:
\begin{equation}
q_{i}(\mathbf{\tilde{v}})=\frac{e^{\tilde{v}_{i}/\zeta _{g_{i}}}}{\sum_{j\in
g_{i}}e^{\tilde{v}_{j}/\zeta _{g_{i}}}}\cdot \frac{e^{\zeta _{g_{i}}\log
\left( \sum_{j\in g_{i}}e^{\tilde{v}_{j}/\zeta _{g_{i}}}\right) }}{\sum_{%
\text{all nests $g$}}e^{\zeta _{g}\log \left( \sum_{j\in g}e^{\tilde{v}%
_{j}/\zeta _{g}}\right) }}.  \label{nested logit}
\end{equation}

The $\mathbf{S}$ function corresponding to a nested logit model is
\begin{equation}  \label{S nested logit}
S_i(\mathbf{q})=q_i^{\zeta_{g_i}} \left(\sum_{j\in g_i}
q_j\right)^{1-\zeta_{g_i}}
\end{equation}
Using this, and applying Proposition \ref{equivalence}, the nested logit
choice probabilities (\ref{nested logit}) are also equivalent to those from
a GERI model with valuations
\begin{equation}
{v}_{i}=\tilde{v}_{i}-\zeta _{g_{i}}\log p_{i}^{0}-(1-\zeta _{g_{i}})\log
\left( \sum_{j\in g_{i}}p_{j}^{0}\right) ,\quad i\in \left\{ 1,\ldots
,n\right\} .  \label{nested-v}
\end{equation}%
The $\mathbf{S}$ function for the nested logit model in Eq. (\ref{S nested
logit}) has several interesting features, relative to the Shannon entropy.
First, Eq. (\ref{S nested logit}) allows us to write the generalized entropy
$\Omega _{\mathbf{S}}(\mathbf{p})$ as%
\begin{equation}
\Omega_{\mathbf{S}} (\mathbf{p})=-\sum_{i=1}^{N}\zeta _{g_{i}}p_{i}\log
p_{i}-\sum_{i=1}^{N}(1-\zeta _{g_{i}})p_{i}\log \left( \sum_{j\in
g_{i}}p_{j}\right) .  \label{Entropy nested logit}
\end{equation}%
The first term in Eq (\ref{Entropy nested logit}) captures the Shannon
entropy within nests, whereas the second term captures the information
between nests. According to this, we may interpret Eq. (\ref{Entropy nested
logit}) as an \textquotedblleft augmented\textquotedblright\ version of
Shannon entropy.

Second, 
when the nesting parameter $\zeta _{g_{j}}=1$, then $\mathbf{S}$ is the
identity function ($S_{j}(\mathbf{p})=p_{j}$ for all $j$), corresponding to
the Shannon entropy. When $\zeta _{g_{j}}<1$, then $S_{j}(\mathbf{p})\geq
p_{j}$; here, $\mathbf{S}(\mathbf{p})$ behaves as a probability weighting
function which tends to overweight options $j$ belonging to larger nests. At
the extreme $\zeta _{g_{j}}\rightarrow 0$, all options within the same nest
effectively collapse into one aggregate option and become perfect
substitutes.

From the discrete choice perspective, nested logit choice probabilities
allow for correlation in the utility shocks ($\epsilon $'s) corresponding to
the different choice options. Analogously, in an information cost function
constructed from the nested logit $\mathbf{S}$ function in Eq. (\ref{S
nested logit}), there will be a common cost component across all options
belonging to the same nest, corresponding to the term $(\sum_{j\in
g_{j}}p_{j})^{1-\zeta _{g_{j}}}$ which is common to all $S_{j}(\mathbf{p})$
for $j\in g_{j}$. From an information processing perspective, this suggests
that there are spillovers in gathering information for options in the same
nest. Information spillovers across choices arise in many decision
environments. For example, a supermarket shopper gains information about common
features of the vegetables, such as the average freshness and quality, while
looking at any of them. In animal foraging, animals who have
information about presence of predators in one grazing site also use that
information to update about predator presence at other nearby sites.

For the Shannon entropy, in contrast, these common terms do not exist, so
that there are no spillovers across options in information processing. From
a behavioral point of view, then, more correlated utility shocks makes each
option's signal harder to distinguish -- there is more redundant information
-- implying that multinomial logit choice probabilities, which would ignore
this correlation, manifest a type of correlation neglect.

To illustrate this point, we compute a GERI model using the nested-logit
cost function. (This requires solving the fixed point equation (\ref%
{GERI_fixedpoint}).) In this example, there are five options, in which the
valuations $\mathbf{v}=(v_{1},v_{2},\ldots ,v_{5})^{\prime }$ are drawn
i.i.d. uniformly from the unit interval. We assume that options (1,2,3) are
in one nest, and options (4,5) are in a second nest. With this
specification, all five options are \emph{a priori} identical, and have
equal probability of being the option with the highest valuation. Hence, we
might expect that any non-uniform choice probabilities should reflect
underlying asymmetries in the information cost function.

In Table \ref{nestedlogittable}, we report the average choice probability
for each option according for several specifications of the nested logit
cost function. In the top panel, we set $\zeta _{1}=\zeta _{2}=1$,
corresponding to the multinomial logit model. In the bottom panel, we set $%
\zeta _{1}=\zeta _{2}=0.5$.

\begin{table}[htbp]
\begin{center}
\begin{tabular}{|l||ccccc|}
\hline
Choice probs: & Option 1 & Option 2 & Option 3 & Option 4 & Option 5 \\
\hline
& \multicolumn{5}{c|}{\emph{Multinomial logit: $\zeta_1=1,\ \zeta_2=1$}} \\
Avg: & 0.200 & 0.200 & 0.200 & 0.200 & 0.200 \\
Median: & 0.194 & 0.194 & 0.194 & 0.194 & 0.194 \\
Std: & 0.060 & 0.060 & 0.060 & 0.060 & 0.060 \\ \hline
Overall efficiency: & \multicolumn{5}{c|}{Pr(\text{Choosing the best option}%
) = 0.283} \\
&  &  &  &  &  \\ \hline
& \multicolumn{5}{c|}{\emph{Nested logit: $\zeta_1=0.5,\ \zeta_2=0.5$}} \\
Avg: & 0.221 & 0.221 & 0.221 & 0.169 & 0.169 \\
Median: & 0.200 & 0.200 & 0.200 & 0.157 & 0.157 \\
Std: & 0.116 & 0.116 & 0.116 & 0.081 & 0.081 \\ \hline
Overall efficiency: & \multicolumn{5}{c|}{Pr(\text{Choosing the best option}%
) = 0.355} \\ \hline
\end{tabular}%
\end{center}
\caption{Choice Probabilities in GERI model: Nested Logit vs. Multinomial
Logit}
\label{nestedlogittable}
\end{table}

As we expect, we see that the average choice probabilities are identically
equal to 0.2 across all five options in the multinomial logit case. As we
remarked before, this reflects the feature of the Shannon-based information
cost function ($S_{i}(\mathbf{p})=p_{i}$) in which information costs are
separable across all five choices.\footnote{%
In the nested logit case, we obtained the unconditional distribution by
iterating over the fixed point relation $\mathbf{p}^{0}=\mathbb{E}\mathbf{p}(%
\mathbf{V})$, starting from the multinomial logit distribution.}
Unlike the multinomial logit case, we see that choice probabilities are
higher for the options 1,2 and 3, which constitute the larger nest, and
smaller for options 4,5 which constitute the smaller nest. (However, within
nest, the choice probabilities are identical.) The non-uniform choice
probabilities for the nested logit model reflect the cost spillovers across
options in the structure of the nested logit information cost function.

Moreover, the performance of the two models is surprisingly different. Under
the multinomial logit specification, the overall efficiency -- defined as
the average probability of choosing the option with the highest valuation --
is 28\%. The overall efficiency for the nested logit, however, is higher,
being over 35\%.

This simple example demonstrates the substantive importance of allowing for
information cost functions beyond the Shannon entropy, which leads to
multinomial logit choice probabilities. Obviously, it makes a difference for
a DM to be processing information using the nested logit cost function, as
compared to the Shannon cost function, as the highest valuation option is
chosen with higher probability on average using the nested logit cost
function. 

\section{Summary}

The central result in this paper is the observational equivalence between a
random utility discrete choice model and a corresponding Generalized Entropy
Rational Inattention (GERI) model. Thus the choice probabilities of any
additive random utility discrete choice model can be viewed as emerging from
rationally inattentive behavior, and vice-versa; we can go back and forth
between the two paradigms.\footnote{%
In a similar vein, \citet{Webb16} demonstrates an equivalence between random
utility models and bounded-accumulation or drift-diffusion models of choice
and reaction times used in the neuroeconomics and psychology literature.}
Then, in order to apply an additive random utility discrete choice model, it
is no longer necessary to assume that decision makers are completely aware
of the valuations of all the available options. This is important, as it is
clearly unrealistic to expect that decision makers to be aware of all
options in a large set of options.


The underlying idea is that, by exploiting convex analytic properties of the
discrete choice model, we show a \textquotedblleft
duality\textquotedblright\ between the discrete choice and GERI models in
the sense of convex conjugacy. Precisely, the surplus function of a discrete
choice model has a convex conjugate that is a generalized entropy.
Succinctly, then, GERI models are rational inattention problems in which the
information cost functions are built from the convex conjugate functions of
some additive random utility discrete choice model.

A few remarks are in order. First, the equivalence result in this paper is
at the individual level, hence it also holds for additive random utility
models with random parameters, including the mixed logit or random
coefficient logit models which have been popular in applied work.\footnote{%
See, for instance, \citet{BLP1995}, \citet{McFaddenTrain00}, %
\citet{FoxEtAl12}.} Any mixed discrete choice model such as these is
observationally equivalent to a mixed GERI model.

In addition, there is also a connection between the results here and papers
in the decision theory literature. The GERI optimization problem (\ref%
{GERI_program}) bears resemblance to the variational preferences that %
\citet{MaccheroniEtal06} propose to represent ambiguity averse preferences,
as well as to the revealed perturbed utility paradigm proposed by %
\citet{FudenbergEtAl16} to model stochastic choice behavior.
\citet{GulNatenzonPesendorfer14} shows an equivalence between random utility
and an \textquotedblleft attribute rule\textquotedblright\ model of
stochastic choice. The main point in this paper is to establish a duality
between rational inattention models and random utility discrete choice
models, which results in observational equivalence of their choice
probabilities. A similar duality might arise between random utility discrete
choice models and these other models from decision theory.

Finally, there are rational inattention models outside the GERI framework;
that is, rational inattention models with information cost functions outside
the class of generalized entropy functions introduced in this paper.\footnote{
As an example, the function
$g(\mathbf{p})=-\sum_{i=1}^N \log(p_i)$ 
is not a generalized entropy function; thus a rational inattention model using this as an information cost function would lie outside the GERI framework.}
Obviously, choice probabilities from these non-GERI models would not be
equivalent to those which can be generated from random utility
discrete-choice models; it will be interesting to examine the empirical
distinctions that non-GERI choice probabilities would have.

\newpage
\phantomsection
\addcontentsline{toc}{section}{References} \singlespacing
\setlength{\bibsep}{0.1pt}

\bigskip
\appendix\singlespace

\section{Proofs and additional results}

\textbf{Notation.} Vectors are denoted simply as $\mathbf{q}=\left(
q_{1},...,q_{N}\right) $. A univariate function applied to a vector is
understood as coordinate-wise application of the function, e.g., $e^{\mathbf{%
q}}=\left( e^{q_{1}},...,e^{q_{N}}\right) $. Consequently, if $a$ is a real
number then $a+\mathbf{q}=\left( a+q_{1},...,a+q_{J}\right) $.
The gradient with respect to a vector $\tilde{\mathbf{v}}$ is $\nabla _{%
\mathbf{\tilde{v}}}$; e.g., for $\tilde{\mathbf{v}}=\left(
v_{1},...,v_{N}\right) $, $\nabla _{\mathbf{\tilde{v}}}W\left( \tilde{%
\mathbf{v}}\right) =\left( \frac{\partial W\left( \mathbf{\tilde{v}}\right)
}{\partial \tilde{v}_{1}},...,\frac{\partial W\left( \mathbf{\tilde{v}}%
\right) }{\partial \tilde{v}_{N}}\right) $. The Jacobian is denoted $J$
with, for example,
\begin{equation*}
J_{\log \mathbf{S}}\left(\mathbf{q}\right) =\left(
\begin{array}{ccc}
\frac{\partial \log S_{1}(\mathbf{q})}{\partial q_{1}} & ... & \frac{%
\partial \log S_{1 }(\mathbf{q})}{\partial q_{N}} \\
... & ... & ... \\
\frac{\partial \log S_{N }(\mathbf{q})}{\partial q_{1}} & ... & \frac{%
\partial \log S_{ N}(\mathbf{q})}{\partial q_{N}}%
\end{array}%
\right) .
\end{equation*}%
A dot indicates an inner product or products of vectors and matrixes. For a
vector $\mathbf{q}$, we use the shorthand $\mathbf{1}\cdot \mathbf{q}=\sum_i
q_i$. The unit simplex in $\mathbb{R}^{N}$ is $\Delta $.

\label{Proofs}

\bigskip

\begin{proof}[Proof of proposition \protect\ref{General_Entropy}]
We first evaluate $W^{\ast }\left( \mathbf{q}\right) $. If $\mathbf{1}\cdot
\mathbf{q}\neq 1$, then
\begin{equation*}
\mathbf{q}\cdot \left(\tilde{\mathbf{v}}+\gamma \right) -W\left( \tilde{%
\mathbf{v}}+\gamma \right) =\mathbf{q}\cdot \tilde{\mathbf{v}}-W\left(
\tilde{\mathbf{v}}\right) +\left( \mathbf{1}\cdot \mathbf{q}-1\right) \gamma
,
\end{equation*}%
which can be made arbitrarily large by changing $\gamma $ and hence $W^{\ast
}\left( \mathbf{q}\right) =\infty $. Next consider $\mathbf{q}$ with some $%
q_{j}<0$. $W\left( \tilde{\mathbf{v}}\right) $ decreases towards a lower
bound
as $v_{j}\rightarrow -\infty
$. Then $\mathbf{q}\cdot \tilde{\mathbf{v}}-W\left( \tilde{\mathbf{v}}%
\right) $ increases towards $+\infty $ and hence $W^{\ast }$ is $+\infty $
outside the unit simplex $\Delta $.

For $\mathbf{q}\in \Delta $, we solve the maximization problem
\begin{equation}
W^{\ast }(\mathbf{q})=\sup_{\tilde{\mathbf{v}}}\{\mathbf{q}\cdot \tilde{%
\mathbf{v}}-W(\tilde{\mathbf{v}})\}.  \label{G*}
\end{equation}%
Note that for any constant $k$ we have $W(\tilde{\mathbf{v}}+k\cdot \mathbf{1%
})=k+W(\tilde{\mathbf{v}})$, so that we normalize $\mathbf{1}\cdot \tilde{%
\mathbf{v}}=0$. Maximize then the Lagrangian $\mathbf{q}\cdot \tilde{\mathbf{%
v}}-W\left( \tilde{\mathbf{v}}\right) -\lambda \left( \mathbf{1}\cdot \tilde{%
\mathbf{v}}\right) $ with first-order conditions $0=q_{j}-\frac{\partial
W\left( \tilde{\mathbf{v}}\right)}{\partial\tilde{v}_j} -\lambda $, which
lead to $\lambda =0$. Then%
\begin{eqnarray*}
\mathbf{q} &=&\nabla _{\tilde{\mathbf{v}}}W\left( \tilde{\mathbf{v}}\right)
\Leftrightarrow \\
\mathbf{q}e^{W\left( \tilde{\mathbf{v}}\right) } &=&\nabla _{\tilde{\mathbf{v%
}}}\left( e^{W\left( \tilde{\mathbf{v}}\right) }\right) =\mathbf{H}\left( e^{%
\tilde{\mathbf{v}}}\right) \Leftrightarrow \\
\mathbf{S}\left( \mathbf{q}\right) e^{W\left( \tilde{\mathbf{v}}\right) }
&=&e^{\tilde{\mathbf{v}}}\Leftrightarrow \\
\log \mathbf{S}\left( \mathbf{q}\right) +W\left( \tilde{\mathbf{v}}\right)
&=&\tilde{\mathbf{v}}\Rightarrow \\
\mathbf{q}\cdot \log \mathbf{S}\left( \mathbf{q}\right) +W\left( \tilde{%
\mathbf{v}}\right) &=&\mathbf{q}\cdot \tilde{\mathbf{v}}.
\end{eqnarray*}%
Inserting this into (\ref{G*}) leads to the desired result.

$W$ is convex and closed and hence $W$ is the convex conjugate of $W^{\ast }$
\citep[][Thm. 12.2]{Rockafellar1970}. This, along with Fenchel's equality %
\citep[Thm. 23.5]{Rockafellar1970}, proves part (iii). Finally, for part
(i), let $\mathbf{q}$ be a solution to problem (\ref{surplusconjugate}).
Then, by the homogeneity of $\mathbf{H}$ we have $\mathbf{q}=\frac{1}{\alpha
}\mathbf{H}(e^{\tilde{\mathbf{v}}})$, where $\alpha =\sum_{j=1}^{N}H_{j}(e^{%
\tilde{\mathbf{v}}})$. Then, by the definition of $\mathbf{S}$ it follows
that $\mathbf{S}(\mathbf{q})=\frac{e^{\tilde{\mathbf{v}}}}{\alpha }$.
Replacing the latter expression in Eq. (\ref{surplusconjugate}) we get
\begin{eqnarray*}
W(\tilde{\mathbf{v}}) &=&\mathbf{q}\tilde{\mathbf{v}}-\mathbf{q}\log \left(
e^{\tilde{\mathbf{v}}}/\alpha \right) , \\
&=&\mathbf{q}\tilde{\mathbf{v}}-\mathbf{q}\left( \log e^{\tilde{\mathbf{v}}%
}+\log \alpha \right) , \\
&=&\log \left( \sum_{j=1}^{N}H_{j}(e^{\tilde{\mathbf{v}}})\right) .
\end{eqnarray*}

\end{proof}

\bigskip

\begin{proof}[Proof of Proposition \protect\ref{Generator}]
Continuity of $\mathbf{S}$ follows from continuity of the partial
derivatives of $W$, which is immediate from the definition. Homogeneity of $%
\mathbf{S}$ is equivalent to homogeneity of $\mathbf{H}$. Using the
homogeneity property of $W$
\begin{equation*}
\mathbf{S}^{-1}(\lambda e^{\tilde{\mathbf{v}}})=\nabla_{\mathbf{v}}(e^{W(%
\tilde{\mathbf{v}}+\log \lambda)})=\lambda \nabla_{\mathbf{v}}(e^{W(\tilde{%
\mathbf{v}})})=\lambda\mathbf{S}^{-1}(e^{\tilde{\mathbf{v}}}),
\end{equation*}

\noindent which shows that $\mathbf{H}$ and hence $\mathbf{S}$ are
homogenous of degree 1.\newline

The requirement that $\sum_{i=1}^{N}q_{i}\frac{\partial \log S_{i}(\mathbf{q}%
)}{\partial q_{k}}=1$ in the relative interior of the unit simplex $\Delta $
may be expressed in matrix notation as
\begin{equation*}
(q_{1},\ldots ,q_{N})\cdot J_{\log \mathbf{S}}(\mathbf{q})=(1,\ldots ,1),
\end{equation*}%
where%
\begin{equation*}
J_{\log \mathbf{S}}(\mathbf{q})=\left\{ \frac{\partial \log S_{i}\left(
\mathbf{q}\right) }{\partial q_{j}}\right\} _{i,j=1}^{N}
\end{equation*}%
is the Jacobian of $\log \mathbf{S}(\mathbf{q})$.

Defining $\mathbf{\hat{t}}\equiv\log \mathbf{S}(\mathbf{q})$, we have $%
\mathbf{q}=\mathbf{H}\left( e^{\mathbf{\hat{t}}}\right) $ and hence $W\left(
e^{\mathbf{\hat{t}}}\right) =\log (\mathbf{1}\cdot \mathbf{H}(e^{\mathbf{%
\hat{t}}}))=\log 1=0$ by Proposition \ref{General_Entropy}. Noting that $%
(\log (\mathbf{S}))^{-1}(\mathbf{\hat{t}})=\mathbf{H}(e^{\mathbf{\hat{t}}})$
the requirement in part (ii) is equivalent to
\begin{equation*}
(q_{1},\ldots ,q_{N})=(q_{1},\ldots ,q_{N})\cdot J_{\log \mathbf{S}}(\mathbf{%
q})\cdot J_{(\log \mathbf{S})^{-1}}(\mathbf{\hat{t}})=(1,\ldots ,1)\cdot J_{%
\mathbf{H}(e^{\mathbf{\hat{t}}})}(\mathbf{\hat{t}}).
\end{equation*}

Now, use the {Williams-Daly-Zachary theorem to find that }
\begin{equation*}
(1,\ldots ,1)\cdot J_{\mathbf{H}(e^{\mathbf{\hat{t}}})}(\mathbf{\hat{t}}%
)=\nabla _{\mathbf{\hat{t}}}\left( e^{W\left( \mathbf{\hat{t}}\right)
}\right) =e^{W(\tilde{\mathbf{v}})}\left( q_{1},\ldots q_{N}\right) =\left(
q_{1},\ldots q_{N}\right) .
\end{equation*}%
as required. \bigskip

Part (ii) follows from Proposition~\ref{General_Entropy}(ii).
\end{proof}

\bigskip

\begin{proof}[Proof of proposition \protect\ref{GERI_Solution}]
The Lagrangian for the DM's problem is%
\begin{equation*}
\Lambda =\mathbb{E}\left( \mathbf{V}\cdot \mathbf{A}\right) -\kappa _{%
\mathbf{S}}(\mathbf{p},\mu )+\mathbb{E}\left( \gamma \left( \mathbf{V}%
\right) \left( 1-\sum\limits_{j}p_{j}\left( \mathbf{V}\right) \right)
\right) +\mathbb{E}\left( \sum\limits_{j}\xi _{j}\left( \mathbf{V}\right)
p_{j}\left( \mathbf{V}\right) \right) ,
\end{equation*}%
where $\gamma \left( \mathbf{V}\right) $ and $\xi _{j}\left( \mathbf{V}%
\right) $ are Lagrange multipliers corresponding to condition (\ref%
{constrain2}).

Before we derive the first-order conditions for $p_{j}\left( \mathbf{v}%
\right) $ it is useful to note that we may regard the terms $\log \mathbf{S}%
\left( \mathbf{p}^{0}\right) $ and $\log \mathbf{S}\left( \mathbf{p}\left(
\mathbf{v}\right) \right) $ in the information cost $\kappa _{\mathbf{S}}(%
\mathbf{p},\mu )$ as constant, since their derivatives cancel out by
Proposition~\ref{Generator}(iii). Define $\tilde{v}_{j}=v_{j}+\xi _{j}\left(
\mathbf{v}\right) +\log S_{j}\left( \mathbf{p}^{0}\right) $ and $\mathbf{%
\tilde{v}}=\left( \tilde{v}_{1},...,\tilde{v}_{N}\right) $. Then the
first-order condition for $p_{j}\left( \mathbf{v}\right) $ is easily found
to be
\begin{equation}
\log \mathbf{S}_{j}\left( \mathbf{p}\left( \mathbf{v}\right) \right) =\tilde{%
v}_{j}-\gamma \left( \mathbf{v}\right) .  \label{eq:3}
\end{equation}%
This fixes $\mathbf{p}\left( \mathbf{v}\right) $ as a function of $\mathbf{p}%
^{0}$ since then
\begin{equation}
\mathbf{p}\left( \mathbf{v}\right) =\mathbf{H}\left( e^{\mathbf{\tilde{v}}%
}\right) \exp \left( -\gamma \left( \mathbf{v}\right) \right) .  \label{eq:4}
\end{equation}

If some $p_{j}\left( \mathbf{v}\right) =0$, then we must have $\tilde{v}%
_{j}=-\infty $, which implies that $S_{j}\left( \mathbf{p}^{0}\right) =0$
and the value of $\xi _{j}\left( \mathbf{v}\right) $ is irrelevant. If $%
p_{j}\left( \mathbf{v}\right) >0$, then $\xi _{j}\left( \mathbf{v}\right) =0$%
. We may then simplify by setting $\xi _{j}\left( \mathbf{v}\right) =0$ for
all $j,\mathbf{v}$ at no loss of generality, which means that $\tilde{v}%
_{j}=v_{j}+\log S_{j}\left( \mathbf{p}^{0}\right) $.

Using that probabilities sum to 1 leads to%
\begin{equation*}
\exp \left( \gamma \left( \mathbf{v}\right) \right)
=\sum\limits_{j}H_{j}\left( e^{\mathbf{\tilde{v}}}\right)
\end{equation*}%
and hence (i) follows. Item (ii) then follows immediately.

Now substitute (\ref{GERI_choices}) back into the objective, using $%
p_{j}\left( \mathbf{v}\right) \xi _{j}\left( \mathbf{v}\right) =0$ , to find
that it reduces to%
\begin{equation}
\Lambda =\mathbb{E}\gamma \left( \mathbf{V}\right) =\mathbb{E}\log
\sum\limits_{j}H_{j}\left( e^{\mathbf{\tilde{v}}}\right)  \label{eq:2}
\end{equation}

We may then use (\ref{eq:2}) to determine $\mathbf{p}^{0}$. Now apply Eq. (%
\ref{surplusentropy}) to establish part (iii) of the proposition.
\end{proof}

\bigskip

\begin{proof}[Proof of proposition \protect\ref{thm:consideration}]
Assume, towards a contradiction, that $p_{a}^{0}>0$. Then%
\begin{eqnarray}
p_{a}^{0} &=&\mathbb{E}\left( \frac{H_{a}\left( \left\{ e^{V_c} S_{c}\left(
\mathbf{p}^{0}\right) \right\} _{c=1}^{N}\right) }{\sum\limits_{b}H_{b}%
\left( \left\{ e^{V_c} S_{c}\left( \mathbf{p}^{0}\right) \right\}
_{c=1}^{N}\right) }\right) \\
&<&\mathbb{E}\left( \frac{H_{a}\left( \left\{e^{V_a} S_{c}\left( \mathbf{p}%
^{0}\right) \right\} _{c=1}^{N}\right) }{\sum\limits_{b}H_{b}\left( \left\{
e^{V_a} S_{c}\left( \mathbf{p}^{0}\right) \right\} _{c=1}^{N}\right) }\right)
\label{cminequality} \\
&=&\mathbb{E}\left( \frac{e^{V_a} H_{a}\left( \left\{ S_{c}\left( \mathbf{p}%
^{0}\right) \right\} _{c=1}^{N}\right) }{e^{V_a}\sum\limits_{b}H_{b}\left(
\left\{ S_{c}\left( \mathbf{p}^{0}\right) \right\} _{c=1}^{N}\right) }%
\right) =\mathbb{E}\left( \frac{p^{0}_{a}}{\sum\limits_{b}p^{0}_{b}}\right)
=p^{0}_{a}.  \label{penultimate}
\end{eqnarray}%
The first inequality (\ref{cminequality}) follows from cyclic monotonicity,
which is a property of the gradient of convex functions. (See, for instance, %
\citet[Thm. 23.5]{Rockafellar1970}.) Since the surplus function $W$ is
convex, its gradient, corresponding to the choice probabilities $\mathbf{p}%
(\cdot )$ is a cyclic monotone mapping, implying that
\begin{equation*}
\left[ \mathbf{p}\left( \left\{ e^{v_a} S_{c}\left( \mathbf{p}^{0}\right)
\right\} _{c=1}^{N}\right) -\mathbf{p}\left( \left\{ e^{v_c} S_{c}\left(
\mathbf{p}^{0}\right) \right\} _{c=1}^{N}\right) \right] \cdot \left[
\left\{ e^{v_a} S_{c}\left( \mathbf{p}^{0}\right) \right\}
_{c=1}^{N}-\left\{ e^{v_c} S_{c}\left( \mathbf{p}^{0}\right) \right\}
_{c=1}^{N}\right] \geq 0.
\end{equation*}%
All the terms within the second pair of brackets on the LHS are $\leq 0$,
except for the $a$-th term, which is equal to zero. In order to satisfy the
inequality, then, we must have
\begin{equation*}
p_{a}\left( \left\{ e^{v_a} S_{c}\left( \mathbf{p}^{0}\right) \right\}
_{c=1}^{N}\right) \geq p_{i}\left( \left\{ e^{v_c} S_{c}\left( \mathbf{p}%
^{0}\right) \right\} _{c=1}^{N}\right)
\end{equation*}
with the inequality strict with positive probability. Otherwise,
\begin{equation*}
\sum_{i\neq a}\left\{ p_{i}\left( \left\{ e^{v_a} S_{c}\left( \mathbf{p}%
^{0}\right) \right\} _{c=1}^{N}\right) -p_{i}\left( \left\{ e^{v_c}
S_{c}\left( \mathbf{p}^{0}\right) \right\} _{c=1}^{N}\right) \right\} >0
\end{equation*}
and
\begin{align*}
& \left[ \mathbf{p}\left( \left\{ e^{v_a} S_{c}\left( \mathbf{p}^{0}\right)
\right\} _{c=1}^{N}\right) -\mathbf{p}\left( \left\{ e^{v_c} S_{c}\left(
\mathbf{p}^{0}\right) \right\} _{c=1}^{N}\right) \right] \cdot \left[
\left\{ e^{v_a} S_{c}\left( \mathbf{p}^{0}\right) \right\}
_{c=1}^{N}-\left\{ e^{v_c} S_{c}\left( \mathbf{p}^{0}\right) \right\}
_{c=1}^{N}\right] \\
=& \sum_{c\neq a}\left[ {p}_{c}\left( \left\{ e^{v_a} S_{c}\left( \mathbf{p}%
^{0}\right) \right\} _{c=1}^{N}\right) -{p}_{c}\left( \left\{ e^{v_c}
S_{c}\left( \mathbf{p}^{0}\right) \right\} _{c=1}^{N}\right) \right] \left[
e^{v_a} -e^{v_c} \right] S_{c}\left( \mathbf{p}^{0}\right) \\
\leq & \max_{c\neq a}\left[ \left( e^{v_a} -e^{v_c} \right) S_{c}\left(
\mathbf{p}^{0}\right) \right] \sum_{c\neq a}\left[ {p}_{c}\left( \left\{
e^{v_a} S_{c}\left( \mathbf{p}^{0}\right) \right\} _{c=1}^{N}\right) -{p}%
_{c}\left( \left\{ e^{v_c} S_{c}\left( \mathbf{p}^{0}\right) \right\}
_{c=1}^{N}\right) \right] \leq 0
\end{align*}%
with the final inequality strict with positive probability. Hence, we
conclude that $p^{0}_{a}=0$.
\end{proof}

\bigskip In the case of the Shannon entropy, Corollary \ref%
{thm:consideration} can be strengthened considerably. In that case, any
alternative that is dominated by another alternative in all states of the
world will never be chosen, as shown in the following corollary:

\begin{corollary}
\label{thm:considerationShannon} Let $S$ be the identity. Suppose that
option $a$ is dominated by option $d$ in the sense that $\forall \mathbf{v}%
\in \mathcal{V}:v_{a}\leq v_{d}$ with strict inequality for some $\mathbf{v}$%
. Then $p_{a}^{0}=0$.
\end{corollary}

\begin{proof}
Suppose to get a contradiction that $p_{a}^{0}>0$. From (\ref%
{Generalized_Logit}), obtain that for all options $a$%
\begin{equation*}
1=\frac{p^0_a}{p^0_a}=\frac{1}{p^0_a}\mathbb{E}p_a(\mathbf{V})= \mathbb{E}%
\left( \frac{\exp \left( V_{a}\right) }{\sum\limits_{b}\exp \left(
V_{b}\right) p^{0}_{b}}\right).
\end{equation*}%
Then
\begin{equation*}
\mathbb{E}\left( \frac{\exp \left( V_{d}\right) }{\sum\limits_{b}\exp \left(
V_{b}\right) p^{0}_{b}}\right) >1,
\end{equation*}%
which is a contradiction.
\end{proof}

\bigskip

\begin{proof}[Proof of Proposition \protect\ref{niceproperties}]
\emph{Independence:} By independence, we have, for all $i$, $p_i(\mathbf{v}%
)=k_i$, a constant. Then $p^0_i=k_i$ and $\kappa_{\mathbf{S}}(\mathbf{p}%
(\cdot),\mu)=0$.

\emph{Convexity:} 
Consider two sets of choice probabilities $\mathbf{p}_{1}\left( \mathbf{v}%
\right) ,\mathbf{p}_{2}\left( \mathbf{v}\right) ,\mathbf{v}\in \mathcal{V}$,
where both have the same implied unconditional probabilities $\mathbb{E}%
\mathbf{p}_{1}(\mathbf{V})=\mathbb{E}\mathbf{p}_{2}(\mathbf{V})$. For $\rho
\in \lbrack 0,1]$, define $\mathbf{p}_{\rho }$ as the convexification $\rho
\mathbf{p}_{1}\left( \mathbf{v}\right) +\left( 1-\rho \right) \mathbf{p}%
_{2}\left( \mathbf{v}\right) $. Then we would like to show that
\begin{equation*}
\rho \kappa_{\mathbf{S}} \left( \mathbf{p_{1}}(\cdot),\mu \right) +\left(
1-\rho \right) \kappa_{\mathbf{S}} \left( \mathbf{p_{2}}(\cdot),\mu \right)
\geq \kappa \left( \mathbf{p_{\rho }}(\cdot),\mu \right) .
\end{equation*}%
But
\begin{align*}
& \rho \kappa_{\mathbf{S}} \left( \mathbf{p_{1}}(\cdot),\mu \right) +\left(
1-\rho \right) \kappa_{\mathbf{S}} \left( \mathbf{p_{2}}(\cdot)\,\mu \right)
-\kappa \left( \mathbf{p_{\rho }}(\cdot),\mu \right) \\
=& -\rho \Omega_{\mathbf{S}} \left( \mathbf{p}_{1}\right) -\left( 1-\rho
\right) \Omega_{\mathbf{S}} \left( \mathbf{p}_{2}\right) +\Omega_{\mathbf{S}%
} \left( \rho \mathbf{p}_{1}+\left( 1-\rho \right) \mathbf{p}_{1}\right) ,
\end{align*}%
which is positive by concavity of $\Omega_{\mathbf{S}} \left( \mathbf{p}%
\right) $ (Proposition~\ref{Generator}(ii)).
\end{proof}

\bigskip

\begin{lemma}
\label{h is invertible}$\mathbf{H}$ is invertible.
\end{lemma}

\begin{proof}[Proof of Lemma \protect\ref{h is invertible}]
We shall make use of Ruzhansky and Sugimoto's %
\citeyear{RuzhanskySugimoto2015} invertibility result applied to $\mathbf{H}$%
. The Jacobian of $\tilde{\mathbf{v}}\rightarrow \mathbf{H}\left( e^{\tilde{%
\mathbf{v}}}\right) $ is $\left\{ e^{W\left( \tilde{\mathbf{v}}\right) }%
\frac{\partial W\left( \tilde{\mathbf{v}}\right)}{\partial v_i}\frac{%
\partial W\left( \tilde{\mathbf{v}}\right)}{\partial v_j} \right\} +\left\{
e^{W\left( \tilde{\mathbf{v}}\right) }\frac{\partial^2 W\left( \tilde{%
\mathbf{v}}\right)}{\partial v_i \partial v_j} \right\} $. The first matrix
is positive definite since all choice probabilities are positive, the second
matrix is positive semidefinite due to the convexity of $W$, hence this
matrix is everywhere positive definite and then the Jacobian determinant of $%
\tilde{\mathbf{v}}\rightarrow \mathbf{H}\left( e^{\tilde{\mathbf{v}}}\right)
$ never vanishes. This implies in turn that the Jacobian determinant of the
composition $\mathbf{y}\rightarrow \log \mathbf{y}\rightarrow \mathbf{H}%
\left( \mathbf{y}\right) $ never vanishes. It remains to show that $\inf_{%
\mathbf{y}\in \Delta }\left\Vert \mathbf{H}\left( \mathbf{y}\right)
\right\Vert >0$. But $\mathbf{y}\in \Delta $ implies that
\begin{eqnarray*}
\left\Vert \mathbf{H}\left( \mathbf{y}\right) \right\Vert &=&e^{W\left( \log
\mathbf{y}\right) }\left\Vert \nabla W\left( \log \mathbf{y}\right)
\right\Vert \\
&\geq &e^{\mathbb{E}\max_{j}\left\{ \log y_{j}+\varepsilon _{j}\right\}
}J^{-1/2} \\
&\geq &e^{\max_{j}\left\{ \log y_{j}+\mathbb{E}\varepsilon _{j}\right\}
}J^{-1/2} \\
&=&\max_{j}\left\{ y_{j}e^{\mathbb{E}\varepsilon _{j}}\right\} J^{-1/2} \\
&\geq &\left\Vert \left( y_{1}e^{\mathbb{E}\varepsilon _{1}},...,y_{J}e^{%
\mathbb{E}\varepsilon _{N}}\right) \right\Vert J^{-1} \\
&\geq &\left( \sum\limits_{j=1}^{N}e^{-2\mathbb{E}\varepsilon _{j}}\right)
^{-1}J^{-1}>0,
\end{eqnarray*}%
where we first used that $\nabla W$ is on the unit simplex, second that the
max operation is convex, third that the sup-norm bounds the euclidean norm,
and fourth that the minimum of $\left\Vert \left( y_{1}e^{\mathbb{E}%
\varepsilon _{1}},...,y_{N}e^{\mathbb{E}\varepsilon _{N}}\right) \right\Vert
$ on the unit simplex is attained at $y_{j}=e^{-2\mathbb{E}\varepsilon
_{j}}\left( \sum\limits_{k=1}^{N}e^{-2\mathbb{E}\varepsilon _{k}}\right)
^{-1},j=1,...,N$.
\end{proof}

\section{Example: Consideration sets and failure of regularity}

\label{nonregularity}

Next, we consider a fully solved out example illustrating the possibility of
zero unconditional choice probabilities and failure of regularity, which can
occur in the rational inattention framework but not in the discrete choice
model, and represent an important point of difference between the two
models. \citet[pp. 293ff]{MatejkaMcKay15} have demonstrated that failures of
regularity can occur in the RI model under Shannon entropy. We show that
such failures also occur in a GERI model, in particular for the nested logit
information cost function introduced in Section~\ref{nestedlogitexample} of
the main text.

Consider a setting with four choice options. Table \ref{tableexample} lists
the valuation vectors for these four options in the three equiprobable
states of the world. We consider both the Shannon and GERI-nested logit
models. (For the nested logit specification, we assume that nest 1 consists
of choices (1,2) with nesting parameter $\zeta _{1}=0.7$, and nest 2
consists of choices (3,4) with parameter $\zeta _{2}=0.8$.

\begin{table}[htbp]
\begin{center}
\begin{tabular}{|l||ccc|}
\hline
State: & $\mathbf{v}^1$ & $\mathbf{v}^2$ & $\mathbf{v}^3$ \\ \hline\hline
Choice 1 & 2 & 3 & 3 \\
Choice 2 & 1 & 2 & 2 \\
Choice 3 & 3 & 1 & 3 \\
Choice 4 & 2 & 4 & 2 \\ \hline
\end{tabular}%
\end{center}
\caption{Valuation vectors in Example 2}
\label{tableexample}
\end{table}

\begin{table}[tbph]
\begin{center}
\begin{tabular}{|c||cccc|}
\hline
Model: & Shannon & Shannon & GERI- & GERI- \\
&  &  & nested logit & nested logit \\ \hline
Choice set: & $\left\{1,2,3\right\}$ & $\left\{1,2,3,4\right\}$ & $%
\left\{1,2,3\right\}$ & $\left\{1,2,3,4\right\}$ \\ \hline
$p^0_1$ & 0.71 & 0.00 & 0.71 & 0.00 \\
$p^0_2$ & 0.00 & 0.00 & 0.00 & 0.00 \\
$p^0_3$ & 0.29 & 0.51 & 0.29 & 0.57 \\
$p^0_4$ & --- & 0.49 & --- & 0.43 \\ \hline
Optimized surplus: &  &  &  &  \\
$\mathbb{E} W(\mathbf{V}+\log\mathbf{S}(\mathbf{p}^0))$ & 2.705 & 2.865 &
4.222 & 6.032 \\ \hline
\end{tabular}%
\end{center}
\caption{Optimal unconditional probabilities for Example 3}
\label{tablepriors}
\end{table}

For each model, we compute the optimal unconditional probabilities (which as
in the previous example, requires solving the fixed-point equation (\ref%
{GERI_fixedpoint})) first for the choice set $\left\{ 1,2,3\right\} $, and
then for the expanded choice set $\left\{ 1,2,3,4\right\} $. This example
illustrates how adding option 4 to the choice set can results in increases
in the choice probabilities of choices (1,2,3) thus showing a failure of the
regularity property. The optimal unconditional probabilities are shown in
Table \ref{tablepriors}. Qualitatively the results are the same between both
the Shannon and GERI-nested logit specifications. With the smaller set of
options, we see that only options 1,2 are chosen with positive
probabilities. When option 4 is added, however, then option 1 drops out of
the consideration set, and only options 3,4 are chosen with positive
probability. This demonstrates a failure of regularity, as the addition of
choice 4 \emph{increases} the prior choice probability for choice 1.
(Moreover, note that with the expanded choice set, option 2 is chosen with
zero probability, even though it is not inferior in all states of the world,
which demonstrates that the characterization of consideration sets in
Corollary \ref{thm:consideration} is not exhaustive.)

Basically, the addition of choice 4 allows agents to form an effective
\textquotedblleft hedge\textquotedblright\ in conjunction with choice 3. In
the state when choice 3 yields a low payoff (state $\mathbf{v}^{2}$), choice
4 yields a high payoff; on the contrary, when choice 4 yields a lower
payoffs (states $\mathbf{v}^{1}$ and $\mathbf{v}^{3}$), choice 3 yields high
payoffs.

\end{document}